\documentclass[12pt]{amsart}
\usepackage[margin=1.25in]{geometry}                
\usepackage{graphicx}
\usepackage{amssymb,amsthm,dsfont}
\usepackage{epstopdf}
\usepackage{setspace}
\usepackage{bm}
\usepackage{color}
\usepackage{natbib}
\usepackage{enumerate}   
\usepackage{hyperref}  
\usepackage{url}
\usepackage{todonotes}
\usepackage{dsfont}
\usepackage{comment, multirow}
\usepackage[foot]{amsaddr}

\theoremstyle{plain}
\newtheorem{thm}{Theorem}[section]

\newtheorem{prop}[thm]{Proposition}

\theoremstyle{definition}

\theoremstyle{remark}

\title[Shapley Value and Counterfactual Simulations]{Group Shapley Value and Counterfactual Simulations in a Structural Model}\thanks{We would like to participants at CeMMAP-Turing Economic Data Science Workshop in May 2024,
Optimization-Conscious Econometrics Conference at the University of Chicago in June 2024, and KBER Summer Institute Conference at Seoul National University in July 2024 for helpful comments and Joonhyuk Lee for excellent research assistance.}
\author[Kwon]{Yongchan Kwon}
\address{Department of Statistics \\
Columbia University \\
New York, NY 10027}
\email{yk3012@columbia.edu}
\author[Lee]{Sokbae Lee}
\address{Department of Economics \\
Columbia University \\
New York, NY 10027}
\email{sl3841@columbia.edu}
\author[Pouliot]{Guillaume A. Pouliot}
\address{Harris School Public Policy \\
University of Chicago \\
Chicago, IL 60637}
\email{guillaumepouliot@uchicago.edu}

\date{October 9, 2024}                                           

\begin{document}

\begin{abstract}
We propose a variant of the Shapley value, the \emph{group Shapley value}, to interpret counterfactual simulations in structural economic models by quantifying the importance of different components. Our framework compares two sets of parameters, partitioned into multiple groups, and applying group Shapley value decomposition yields unique additive contributions to the changes between these sets. The relative contributions sum to one, enabling us to generate an \emph{importance table} that is as easily interpretable as a regression table. The group Shapley value can be characterized as the solution to a constrained weighted least squares problem. Using this property, we develop robust decomposition methods to address scenarios where inputs for the group Shapley value are missing. We first apply our methodology to a simple Roy model and then illustrate its usefulness by revisiting two published papers.
\end{abstract}

\maketitle


\doublespacing


\section{Introduction}


The Shapley value, first introduced by \citet{ShapleyValue}, is a foundational solution concept in cooperative game theory, with desirable properties supported by axiomatic approaches. It is also sustainable in non-cooperative settings \citep[e.g.,][in sequential bargaining]{gul1989bargaining}. \citet{Shorrocks:2013} popularized its use in empirical applications, primarily for decomposing $R^2$ within linear regression models. 
Many studies have used the Shapley-Shorrocks decomposition: see, e.g., \citet{allen2014information}, \citet{allen2014trade}, 
\citet{dustmann2014out} and \citet{henderson2018global} among others.  
Most notably since the work of \citet{LundbergLee}, Shapley values have become one of the main tools of explainable artificial intelligence (AI). They are likewise recognized as useful interpretational tools in other disciplines, such as global sensitivity analysis \citep{Owen:2014}.

We propose to make a use of a particular variant which we call the \emph{group Shapley value}. Specifically, we use it to understand the outcomes of counterfactual simulations in a structural economic model by quantifying the importance of various elements embedded within the model.
In a structural model, the crucial task involves estimating or calibrating the underlying model parameters. Consequently, any changes to these parameter values will directly influence the outcomes of the structural model. Our motivating example is from \citet{DV2019}, who developed a method to disentangle sources of capital misallocation using data from both China and the United States.
Their model includes various components such as shocks to productivity, investment adjustment costs, uncertainty, and so on. They are specified via a number of parameters that are estimated or calibrated: Table~\ref{tab:DV2019:intro} reproduces the parameter estimates from \citet{DV2019}.

\begin{table}[htbp]
{\tiny
\caption{Parameter Estimates from \citet{DV2019}}\label{tab:DV2019:intro}
\begin{center}
\begin{tabular}{ccccccccc}
\hline\hline
 & Production & \multicolumn{2}{c}{Productivity} & Adjustment & Uncertainty & \multicolumn{3}{c}{Other Factors} \\
 & Curvature & Persistence  & Variance &  Costs &  & Correlated & Transitory & Permanent \\
 \hline
Parameter & $\alpha$ & $\rho$ & $\sigma_\mu^2$ & $\xi$ & $V$ & $\gamma$ & $\sigma_\varepsilon^2$ & $\sigma_\chi^2$  \\
\hline
China & 0.71 & 0.91 & 0.15 & 0.13 & 0.10 & $-0.70$ & 0.00 & 0.41 \\
U.S. & 0.62 & 0.93 & 0.08 & 1.38 & 0.03 & $-0.33$ & 0.03 & 0.29 \\
\hline
\end{tabular}
\end{center}
\parbox{6in}{
Notes. This table is reproduced from Tables 1 and 2 in \citet{DV2019}. Information on $\alpha$ can be found on page 2546 in their paper. 
}
}
\end{table}

In general, it would be challenging to compare the relative importance across different components of the parameters because they tend to be defined on totally non-comparable scales. For example, in \citet{DV2019}, the changes in the parameter estimates directly influence the main outcome variable, namely, dispersion in average revenue products of capital ($\Delta \sigma^2_{arpk}$ using their notation); however, they are non-comparable across different components.

\begin{table}[htbp]
{\small
\begin{center}
\centering
\caption{Contributions to Captial Misallocation}\label{tab:DV2019:intro:more}
\begin{tabular}{ccccccccc}
\hline 
Parameter & $\xi$ & $V$ & $\gamma$ & $\sigma_\varepsilon^2$ & $\sigma_\chi^2$  \\
\hline
China & 0.01 & 0.10 & 0.44 & 0.00 & 0.41  \\
U.S. & 0.05 & 0.03 & 0.06 & 0.03 & 0.29 \\
\hline
\end{tabular}
\end{center}
\parbox{5in}{
Notes. This table is reproduced from Table 3 in \citet{DV2019}.
}
}
\end{table}

One common practice in economics is to apply the principle of \emph{ceteris paribus}, measuring the effect of a component of interest by turning it on (or changing its value), while turning off all the others (or fixing all the others at constant values). 
For example, Table 3 in \citet{DV2019}, which is reproduced here as Table~\ref{tab:DV2019:intro:more}, reports the relative contributions of adjustment costs ($\xi$),
uncertainty ($V$), and other factors ($\gamma$, $\sigma_\varepsilon^2$, and $\sigma_\chi^2$) under the assumption that only the factor of interest is operational. 
This type of evaluation has been popular but the relative contributions do not necessarily sum to one, as acknowledged by \citet{DV2019}.

The ceteris paribus exercise in Table~\ref{tab:DV2019:intro:more} is applied to each country, thus explaining how much individual factors affect each country's capital misallocation. Their method does not directly compare the United States with China; nor does it account for the impact of replacing a particular factor from China to the United States. 
As an alternative, we propose a different decomposition method based on the group Shapley value. 
Our framework needs \emph{two sets of parameters} to compare (e.g., China and the United States in Table~\ref{tab:DV2019:intro}).
Our application of the group Shapley value decomposition generates unique additive contributions for the changes between the two sets of parameters, thereby providing a coherent method for quantitatively evaluating the importance of different components of the model. Decomposition results are summarized in a table, where their relative shares sum to one, allowing for ranking of importance by magnitude.



Generally speaking, there are two modes of sensitivity analysis: global sensitivity analysis \citep[e.g.,][on variance decomposition]{Owen:2014} and local sensitivity analysis \citep[e.g.,][within the generalized method of moments (GMM) framework]{AGS2017}. Our application of the Shapley value is more closely related to the concept of global sensitivity analysis. Although our decomposition may resemble a ``derivative'' to some readers, unlike a regular derivative, the Shapley value can be defined for a non-differentiable ``utility'' function and our setting is not limited to variance decomposition or the GMM framework. 

Although Shapley values are widely used across different disciplines, to the best of our knowledge, they have not been adopted to interpret counterfactual simulations in structural economic models. One might wonder in what sense it would be useful to adopt the framework of the Shapley value in this context. It is well known that the Shapley value is the unique solution that satisfies the four cooperative game theory axioms: Linearity, Dummy, Symmetry, and Efficiency \citep[see, e.g.,][]{moulin2004fair}. When analyzing the impacts of changes in the underlying parameters of a highly nonlinear and complex structural model, where the parameters can be partitioned into multiple groups, it is useful to assign importance of each group to explain the changes in the outcome of the counterfactual simulation. To evaluate these contributions, we may prefer the following properties: separable additive contributions of each group (i.e., Linearity); zero contribution if a change in one group of parameters has no impact at all (i.e., Dummy); identical contributions if two different groups produce identical marginal impacts (i.e., Symmetry); and a sum of all additive contributions that equals the total changes in the outcome (i.e., Efficiency). In other words, if we wish to retain all four axiomatic properties in a decomposition exercise, our proposed method is the only solution that satisfies them. Thus, we regard our approach as an attractive decomposition method based on a strong theoretical foundation. Furthermore, from a practical perspective, we believe our method will help practitioners produce an importance table that is as easily interpretable as a regression table.

The remainder of the paper is organized as follows. 
In Section~\ref{sec:group-shapley}, we provide the definition of the group Shapley value;
in Section~\ref{sec:ls-form}, we show that the group Shapley value can be characterized as a solution to a constrained weighted least squares problem.
Additionally, we consider scenarios where some input for the group Shapley value is missing and propose robust decomposition methods to address this issue.
In Section~\ref{sec:xai}, we describe a canonical use of the Shapley value in the machine learning community.
In Section~\ref{sec:counterfactual}, we provide a novel use of the group Shapley value in order to quantitatively evaluate the different components in counterfactual simulations that are generated by structural models, using a simple Roy model \`{a} la \citet{honore2017poor}.
In Section~\ref{sec:example}, we demonstrate the usefulness of our framework by revisiting \citet{DV2019} and offering new perspectives on contributing factors for capital misallocation.
Specifically, we offer an output table for Shapley value decomposition where we can rank the importance of parameter changes by magnitude.
In Section~\ref{sec:example:more}, we revisit \citet{CGT2016} and use its setting to illustrate the value of our decomposition methods when some input for the group Shapley value is missing. 
We show that our method is effective in providing robust estimates, even when utility computations are costly.
Section~\ref{sec:conclusions} provides the concluding remarks and 
Appendix~\ref{sec:appendix} gives the proof omitted in the main text.

\section{Group Shapley Value}\label{sec:group-shapley}


We start with \citet{gul1989bargaining}'s generalized Shapley value. 
Let $P := \{ 1, 2, \ldots, p \}$ denote a set with $p$ elements. 
Consider a function $g: 2^P \mapsto \mathbb{R}$ that satisfies $g(\emptyset) = 0$. 
Let $\mathcal{P}$ denote the set of all possible partitions of $P$.
For each $\Pi \in \mathcal{P}$ and $M \in \Pi$, 
\citet{gul1989bargaining} defines the generalized Shapley value by 
\begin{align}
\phi_g (M, \Pi) &:=
\sum_{\Psi \subset \Pi} 
\frac{(|\Psi|-1)! (|\Pi| - |\Psi|)!}{|\Pi|!}
\left[ 
g\left( \bigcup_{A \in \Psi} A \right)
-
g\left( \bigcup_{A \in \Psi \setminus \{ M \}} A \right)
\right] \label{gsv} \\
&= \frac{1}{ | \Pi |} \sum_{k=1} ^{| \Pi |}   \underbrace{\frac{1}{\binom{|\Pi|-1}{k-1}} 
 \sum_{\Psi \subset \Pi, |\Psi| =k, M \in \Psi  } \left[ g\left( \bigcup_{A \in \Psi} A \right) - g\left( \bigcup_{A \in \Psi \setminus \{ M \}} A \right) \right]}_{\text{Average change when $M$ is removed from a set $\Psi$ with $|\Psi|=k$.}}. \nonumber
\end{align} 
Here, we use Greek capital letters (e.g., $\Pi$ and $\Psi$) to denote a partition of $P$
and Roman capital letters  (e.g., $M$ and $A$) to denote an element in a partition. 
Furthermore, $|\Psi|$ denotes the number of elements in partition $\Psi$. For ease of notation, we drop the subscript $g$ and use $\phi (M, \Pi)$ to denote $\phi_g (M, \Pi)$ when the context is clear.
Note that the Shapley value  is $\phi (\{i\}, \{ \{j\}: j \in P \})$ with $M = \{i\}$ and
$\Pi = \{ \{j\}: j \in P \}$. 
The generalized Shapley value can also be written as 
\begin{align}\label{gsv-alt}
\phi (M, \Pi) =
\sum_{\Psi \subset \Pi \setminus \{ M \}} 
\frac{|\Psi|! (|\Pi| - |\Psi|-1)!}{|\Pi|!}
\left[ 
g\left( \bigcup_{A \in \Psi \cup \{ M \} } A \right)
-
g\left( \bigcup_{A \in \Psi}  A \right)
\right].
\end{align} 
Two expressions \eqref{gsv} and \eqref{gsv-alt} are identical; \eqref{gsv} is expressed in terms of the marginal changes in the utility from subtracting $M$, whereas \eqref{gsv-alt} is given via the marginal changes in the utility from adding $M$.
Instead of calling $\{ \phi (M, \Pi): M \in \Pi \}$ the generalized Shapley value, we may term it
the \emph{group Shapley value} because it resembles group lasso in variable selection \citep[e.g.,][]{yuan2006model} 
and it may better describe our application of the generalized Shapley value.  
In many economic applications, it would be reasonable to assume a known group structure. 
See Section~\ref{sec:counterfactual} for an example of the group structure in a simple Roy model.

The original Shapley value provides a principled approach to allocating the total utility to individual elements by delivering a unique valuation that satisfies the fairness axioms in cooperative game theory \citep{ShapleyValue}. The fairness axioms mathematically describe desired conditions for a valuation as a functional of a function $g$, namely Linearity, Dummy, Symmetry, and Efficiency.
Extending this original axiomatic characterization, we present a set of axioms for the generalized Shapley value. 
\begin{itemize}
    \item Linearity: $\phi_g$ satisfies the linearity axiom if $\phi_{g_1+ g_2} = \phi_{g_1} + \phi_{g_2}$ for any set functions $g_1$ and $g_2$.
    \item Dummy: $\phi_g$ satisfies the dummy axiom if $g(S\cup A) = g(S)$ for any $S \subseteq \Pi$ implies $\phi_g (A) = 0$ for any set function $g$. 
    \item Symmetry: $\phi_g$ satisfies the symmetry axiom if $g(S\cup A) = g(S \cup B)$ for $S \subseteq \Pi \backslash \{A, B\}$ implies $\phi_g (A) = \phi_g (B)$ for any set function $g$.
    \item Efficiency: $\phi_g$ satisfies the efficiency axiom if $\sum_{M \in \Pi} \phi(M, \Pi) = g(\bigcup_{A \in \Pi}  A)$ for any set function $g$.
\end{itemize}

With these axioms, the following proposition is directly derived from Section 5.5 of \citet{moulin2004fair}.
\begin{prop}
A generalized Shapley value is a unique function that satisfies Linearity, Dummy, Symmetry, and Efficiency axioms.
\end{prop}
A generalized Shapley value gives a unique way to attribute $g(P)$ to individual groups $M \in \Pi$ while satisfying the four axioms above. In other words, no other evaluation function satisfies one or more of these axioms. 

\section{Group Shapley Value via Constrained Weighted Least Squares}\label{sec:ls-form}


The Shapley value has multiple equivalent formulations. 
For example, 
\citet{charnes1988extremal},
\citet{LundbergLee},
and
\citet{aas2021explaining}
express the Shapley value as the optimal solution to a constrained quadratic minimization problem, namely a constrained weighted least squares estimator.
The following proposition is a straightforward application of Theorem 3 in \citet{charnes1988extremal} to our setting.

\begin{prop}[\citet{charnes1988extremal}]\label{opt-result}
The group Shapley value
$\{ \phi (M, \Pi): M \in \Psi \subsetneqq \Pi, M \neq \emptyset \}$ solves the following optimization problem:
\begin{align*}
\min_{ \{\phi_M  : M \in \Psi \subsetneqq \Pi, M \neq \emptyset \} } 
\sum_{\Psi \subsetneqq \Pi}  
\left\{ 
g\left( \bigcup_{M \in \Psi}  M \right)  -  \sum_{M \in \Psi} \phi_M 
\right\}^2 k ( \Pi, \Psi )
\; \text{ subject to }  \; \sum_{M \in \Pi} \phi_M = g \left(  P \right),
\end{align*}
where
\begin{align*}
k ( \Pi, \Psi ) := 
\begin{pmatrix}
|\Pi| - 2 \\
|\Psi| - 1
\end{pmatrix}^{-1}.
\end{align*}
\end{prop}

The proof follows steps similar to those used in \citet{charnes1988extremal}.
However, they did not specifically consider the case of the group Shapley value.
For the purpose of self-containment, we provide details of the proof in Appendix~\ref{sec:appendix}.


While the group Shapley value can be expressed as closed-forms as shown in \eqref{gsv} and \eqref{gsv-alt}, Proposition~\ref{opt-result} yields a new expression because the optimization problem is a weighted linear regression problem with linear constraints. To be more specific, we consider a vector $(\Psi_1, \dots, \Psi_{2^{|\Pi|}-2})$ of subsets of $P$ where $\Psi_i$'s are possible unions of $|\Pi|$'s subset except $\emptyset$ and $P$. We set $\bm{g} = (g(\Psi_1), \dots, g(\Psi_{2^{|\Pi|}-2}))$ and denote a diagonal matrix whose $i$-th diagonal element is $k(\Pi, \Psi_i)$ by $\bm{K}$. Each $\Psi_i$ can be uniquely represented by a $|\Pi|$-dimensional binary vector, and we denote a binary matrix whose $i$-th row represents $\Psi_i$ by $\bm{D}$. Then, the optimization problem in Proposition~\ref{opt-result} can be re-expressed as:
\begin{align*}
    \min_{\bm{\phi}} ( \bm{g} - \bm{D} \bm{\phi})^\top \bm{K} ( \bm{g} - \bm{D} \bm{\phi})
\end{align*}
subject to
\begin{align*}
    \bm{\phi}^\top \mathds{1}_{|\Pi|} = g(P),
\end{align*}
where $\mathds{1}_{|\Pi|}$ is a $|\Pi|$ dimensional vector of ones. Moreover, the solution of this constraint linear regression problem is given as follows:
\begin{align}\label{ls-solutions}
    A^{-1}\left(b -\mathds{1}_{|\pi|} \frac{ \mathds{1}_{|\pi|}^\top A^{-1}b - g(P) }{\mathds{1}_{|\pi|}^\top A^{-1} \mathds{1}_{|\pi|}}\right),
\end{align}
where $A = \bm{D}^\top \bm{K} \bm{D}$ and $b = \bm{D}^\top \bm{K} \bm{g}$.

\subsection{Toy Example}\label{sec:toy}

Suppose that $P = \{1,2,\ldots,5 \}$ and $\Pi = \{ \{1 \}, \{2,3\}, \{4,5\} \}$.
As an example, 
the group Shapley value with $M = \{1\}$ is written as 
\begin{align*}
\phi_g (M, \Pi)
&= \sum_{\Psi \subset \Pi} \frac{(|\Psi|-1)! (|\Pi| - |\Psi|)!}{|\Pi|!} \left[ g\left( \bigcup_{A \in \Psi} A \right) - g\left( \bigcup_{A \in \Psi \setminus \{ M \}} A \right) \right]\\
&= \frac{1}{3} \Bigl( g\left({\{1\}}\right) - g\left( {\emptyset} \right) \Bigr) \\
&+ \frac{1}{3} \Biggl( \frac{\Bigl( g\left({\{1\}} \cup {\{2,3\}} \right) - g\left( {\{2,3\}} \right) \Bigr) + \Bigl( g\left({\{1\}} \cup {\{4,5\}} \right) - g\left( {\{4,5\}} \right) \Bigr) }{2} \Biggr)\\
&+ \frac{1}{3} \Bigl( g\left({\{1\}} \cup {\{2,3\} \cup \{4, 5\}} \right) - g\left( {\{2,3\} \cup \{4, 5\}} \right) \Bigr).
\end{align*}
To write the least squares problem, define
\begin{align*}
\bm{g} := 
\begin{pmatrix}
g ( \{1 \} ) \\
g ( \{2,3 \} ) \\
g ( \{4,5 \} ) \\
g (\{1 \} \cup  \{2,3 \}) \\
g ( \{1 \} \cup  \{4,5 \}) \\
g ( \{2,3 \} \cup  \{4,5 \}) \\
\end{pmatrix}, \;
\bm{D} :=
\begin{pmatrix}
1 & 0 & 0 \\
0 & 1 & 0 \\
0 & 0 & 1 \\
1 & 1 & 0 \\
1 & 0 & 1 \\
0 & 1 & 1 \\
\end{pmatrix}, \;
\bm{\phi} := 
\begin{pmatrix}
\phi_{ \{1 \} } \\
\phi_{ \{2, 3 \} } \\
\phi_{ \{4, 5 \} } 
\end{pmatrix}.
\end{align*}
In this example, 
$k ( \Pi, \Psi ) = 1$ for any $\Psi \subsetneqq \Pi$, implying that $\bm{K}$ is the identity matrix.
Thus, $\bm{\phi}$ minimizes
\begin{align}\label{cls-def}
( \bm{g} - \bm{D} \bm{\phi})^\top  ( \bm{g} - \bm{D} \bm{\phi})
\end{align}
subject to 
$$
\phi_{ \{1 \} } + \phi_{ \{2, 3 \} } + \phi_{ \{4, 5 \} } = g(P).
$$
This is simply a constrained least squares problem, which can be solved easily as shown in \eqref{ls-solutions}. 
However, when $|\Pi|$ gets large, the number of rows of $\bm{D}$ increases very rapidly:
$2^{|\Pi|}-2$.  In the literature \citep[e.g.,][]{LundbergLee}, this difficulty is avoided by sampling the rows of $\bm{D}$
according to the probability distribution induced by the Shapley kernel weights $k ( \Pi, \Psi )$. See Section~\ref{sec:xal-algo} for details.

\subsection{Shapley Bounds and Shapley Minimum Norm Solutions}\label{sec:SB-SMNS}

In practice, it is often challenging to observe every element of the utility vector $\bm{g} = (g(\Psi_1), \dots, g(\Psi_{2^{|\Pi|}-2}))$, especially when the computational cost of each $g$ is expensive. For instance, \citet{CGT2016}, who analyze the contribution of firing cost, tariff rate, and iceberg trade cost to aggregate statistics in a \textit{ceteris paribus} manner, do not provide the entire utility values. We will revisit this example in Section~\ref{sec:example:more}. In such situations, the Shapely value in \eqref{ls-solutions} cannot be calculated. To address this problem, we propose Shapley bounds and Shapley Minimum Norm Solution that infer the Shapley values under user-specified linear constraints.

For $r \in \mathbb{N}$, $A_{\mathrm{const}} \in \mathbb{R}^{r \times |\bm{g}|}$, and $b_{\mathrm{const}} \in \mathbb{R}^r$, we suppose the utility vector $\bm{g}$ satisfies the linear constraints $A_{\mathrm{const}} \bm{g} \leq b_{\mathrm{const}}$. These linear constraints allow researchers to formulate their domain expertise and infer a missing part of $\bm{g}$. Since the Shapley value $\phi$ can be seen as a linear function of $\bm{g}$, we denote it by $\phi(\bm{g})$. Then, for $j \in |\Pi|$, Shapley Upper Bounds (SUB) for the $j$-th element of $\phi$ is defined as follows.
\begin{align*}
\mathrm{max}_{\tilde{\bm{g}} \in \mathcal{G}}  c_j^T \phi(\tilde{\bm{g}}) \quad \mathrm{s.t.} \quad A_{\mathrm{const}} \tilde{\bm{g}} \leq b_{\mathrm{const}},
\tag{SUB}
\end{align*}
where $c_j \in \{0,1\}^{|\Pi|}$ is the $j$-th canonical basis in $\mathbb{R}^{|\Pi|}$ and $\mathcal{G} := \{ \bm{h} \in \mathbb{R}^{|\bm{g}|} : \bm{h}_i = \bm{g}_i $ if $\bm{g}_i$  is observed $\}$. That is, SUB is the maximum possible Shapely value that satisfies given linear constraints. Similarly, we define Shapley Lower Bounds (SLB) as follows.
\begin{align*}
\mathrm{min}_{\tilde{\bm{g}} \in \mathcal{G}}  c_j^T \phi(\tilde{\bm{g}}) \quad \mathrm{s.t.} \quad A_{\mathrm{const}} \tilde{\bm{g}} \leq b_{\mathrm{const}}
\tag{SLB}
\end{align*}

While SUB and SLB can provide informative bounds for the Shapley value, it does not satisfy the efficiency axiom in general, and its value does not sum to $g(P)$. To address this problem, we propose the Shapley Minimum Norm Solution (SMNS) that solves the following optimization problem.
\begin{align*}
\mathrm{min}_{\tilde{\bm{g}} \in \mathcal{G}} \lVert \phi(\tilde{\bm{g}}) - \frac{g(P)}{N} \mathds{1}_{N} \rVert_2 \quad \mathrm{s.t.} \quad A_{\mathrm{const}} \tilde{\bm{g}} \leq b_{\mathrm{const}}
\tag{SMNS},
\end{align*}
where $\mathds{1}_{N}$ is the $N$-dimensional one vector and $\lVert v \rVert_2$ is the $\ell^2$-norm of a vector $v$. 
Here, $\frac{g(P)}{N} \mathds{1}_{N}$ can be seen as the best guess of the Shapley value when the practitioner believes the contribution of every factor is identical, and SMNS infers the Shapely value that is not very different from $\frac{g(P)}{N} \mathds{1}_{N}$ while satisfying the linear constraints.

Note that the solution ($\tilde{g}$) to any of the optimization problems above is a utility. The resulting Shapley value is then obtained by plugging solution $\tilde{g}$ into equation \eqref{ls-solutions}. Either reporting the bounds [SLB, SUB] or using SMNS can be viewed as a reasonable alternative in the presence of missing input. We use both approaches in Section~\ref{sec:example:more}.

\section{Explainable Artificial Intelligence}\label{sec:xai}

The Shapley value and its variations have been deployed in various machine learning applications. As mentioned in the introduction, \cite{LundbergLee} advocate its use in the model interpretation problem (also known as explainable artificial intelligence), where it is employed to attribute a model's prediction to features. \cite{ghorbani2019data} extend its utility to the data valuation problem, introducing a method to quantify the impact of individual data. \cite{pmlr-v151-kwon22a} relax the efficiency axiom and use a semivalue in evaluating data values. \cite{wang2020principled} leverage the Shapley value to quantify the contribution of a local model to a global model in federated learning settings. In addition, the Shapley value has been proposed as a fair valuation method in data marketplaces \citep{tian2022private}. 
For an in-depth exploration of the literature on the Shapley value and its applications in machine learning, we refer the readers to the comprehensive review by \cite{rozemberczki2022shapley}.
In this section, we provide a brief description of an application of Shapley value in the context of explainable artificial intelligence.

For a vector of covariates $x$, 
let $f(x)$ denote a prediction model of a real-valued outcome of interest 
using $x$. 
For $A \subset P$, decompose $x = (x_A, x_{P \setminus A})$.
Here, $x_A$ denotes a subvector of $x$ whose indices are in $A$. For example, if $x = (x_1, x_2, x_3)$ and $A = \{ 1, 2 \}$, $x_A = (x_1, x_2)$.   
For each $x$ and $A$, define 
\begin{align}\label{def-g}
g_{x} ( A  ) := 
\mathbb{E}_{X_{P \setminus A}} [ f(x_A, X_{P \setminus A})   ] - \mathbb{E}_{X} [ f(X)   ],
\end{align}
where the upper case refers to a random vector. 
Note that $g_{x} ( A  )$ is a function of $x_A$.

\subsection{Algorithm for Explainable Artificial Intelligence}\label{sec:xal-algo}

\begin{enumerate}[(i)]
\item Choose $x=x^\ast$, where $x^\ast$ could be an observed vector of an individual we would like to study.

\item Predict $f(x^\ast)$ using a machine-learning prediction method.

\item Specify $\Pi$. 

\item Sample $q$ rows of $\bm{D}$ according to the probability distribution induced by the Shapley kernel weights $k ( \Pi, \Psi )$, where $q \gg |\Pi|$. 

\item For each row of $\bm{D}$, compute a sample analog of \eqref{def-g}.

\item Solve the constrained least squares problem in \eqref{cls-def}.

\end{enumerate}


We end this section by commenting that 
there have been some debates on causal interpretations of the Shapley values for explainable AI, in particular, when the covariates are dependent on each other.
See, e.g.,  
\citet{LundbergLee},
\citet{aas2021explaining},
\citet{janzing2020feature},
and
\citet{eskes2020causal}
among others.
This debate is irrelevant for our application of Shapley values to counterfactual simulations.



\section{Counterfactual Simulations}\label{sec:counterfactual}

We propose to use Shapley values for evaluating the different components in counterfactual simulations that are generated by structural models. 

\subsection{A Simple Roy Model \'{a} la \citet{honore2017poor}}
To be explicit, we consider a simple Roy model used in \citet{honore2017poor}. 
In their setup, there are two sectors:  $s \in \{1, 2\}$. Worker $i$ earns sector-specific income
$w_{si1}$ in period 1 in the following form:
\begin{align*}
\log (w_{si1}) = x_{si}^\top \beta_s + \varepsilon_{si1}, 
\end{align*}
where $x_{si}$ is a vector of sector-specific human capital,
$\beta_s$ is a vector of sector-specific parameters,
and $\varepsilon_{si1}$ is a sector-specific unobserved random variable in period 1.
Sector-specific income
$w_{si2}$ in period 2 is generated by
\begin{align*}
\log (w_{si2}) = x_{si}^\top \beta_s + 1 \{ d_{i1} = s \} \gamma_s + \varepsilon_{si2}, 
\end{align*}
where $d_{i1}$ is the sector chosen in period 1,
$\gamma_s$ is a sector-specific parameter that represents the premium of staying in the same sector in period 2, 
and $\varepsilon_{si1}$ is a sector-specific unobserved random variable in period 2.
The unobserved random variables $(\varepsilon_{1it}, \varepsilon_{2it})$
are assumed to be bivariate normally distributed
with mean zeros, variances $(\sigma_1^2, \sigma_2^2)$ and correlation $\tau$,
 and i.i.d. over time.

Workers maximize discounted income. In time period 2, 
$d_{i2} = 1$ and the resulting income  is $w_{i2} = w_{1i2}$ if
\begin{align*}
x_{1i}^\top \beta_1 + 1 \{ d_{i1} = 1 \} \gamma_1 + \varepsilon_{1i2}
>
x_{2i}^\top \beta_2 + 1 \{ d_{i1} = 2 \} \gamma_2 + \varepsilon_{2i2},
\end{align*}
and $d_{i2} = 2$ and $w_{i2} = w_{2i2}$ otherwise.
In time period 1,  $d_{i1} = 1$ if and only if
\begin{align*}
w_{1i1} + \rho \mathbb{E} \left[
\max \{ w_{1i2}, w_{2i2} \} | x_{1i}, x_{2i}, d_{i1} = 1
\right]
>
w_{2i1} + \rho \mathbb{E} \left[
\max \{ w_{1i2}, w_{2i2} \} | x_{1i}, x_{2i}, d_{i1} = 2
\right],
\end{align*}
where $\rho$ is discount factor.

Let $\theta = (\beta_1, \beta_2, \gamma_1, \gamma_2, \sigma_1^2, \sigma_2^2, \tau, \rho)$ denote
the vector of all structural parameters. 
Let $\theta^b$ denote the benchmark vector of parameter values, for example, estimated values,
and $\theta^c$ a counterfactual vector of parameter values.
If we write $(\theta^{c}_{A}, \theta^b_{P \setminus A})$, we mean the combination of parameter values such
that the $A$ subset of $\theta$ is set at the counterfactual values, while setting the $P \setminus A$ subset of 
$\theta$ at the benchmark values.
Let $W = W(\theta)$ denote the observed random variables generated by the model given $\theta$. In the Roy model example, we have that 
$W = (w_{i1}, w_{i2}, d_{i1}, d_{i2}, x_{i1}, x_{i2})$.
Let $f(\theta)$ be the real-valued quantity of interest in a counterfactual simulation
that can be computed by simulating $W$ given $\theta$. 
For example, $f(\theta)$ be a measure of the changes in income inequality from period 1 to period 2. 
Finally, we let 
\begin{align}\label{g-counterfactual}
g( A ) :=  f(\theta^{c}_{A}, \theta^b_{P \setminus A})
- f(\theta^{b}). 
\end{align}

\subsection{Algorithm for Counterfactual Simulations}

\begin{enumerate}
\item Choose $\theta = \theta^b$, where $\theta_b$ could be estimated parameter values using the dataset.

\item Simulate $W$ given $\theta^b$ and evaluate $f(\theta^{b})$.

\item Specify $\Pi$. 

\item Sample $q$ rows of $\bm{D}$ according to the probability distribution induced by the Shapley kernel weights $k ( \Pi, \Psi )$, where $q \gg |\Pi|$. 

\item For each row of $\bm{D}$, 
simulate $W$ given 
$(\theta^{c}_{A}, \theta^b_{P \setminus A})$
and evaluate $g(A) =  f(\theta^{c}_{A}, \theta^b_{P \setminus A})$.

\item Solve the constrained least squares problem in \eqref{cls-def}.

\end{enumerate}

\subsection{An Example: $\Pi$, $\theta^b$, $\theta^c$ and $f(\theta)$}

Suppose that $\tau = 0$ and $\rho = 0.95$ are fixed, as in \citet{honore2017poor}.
Then, 
$P = \{ \beta_1, \beta_2, \gamma_1, \gamma_2, \sigma_1^2, \sigma_2^2 \}$.
Suppose that 
$$
\Pi = \{ 
\{\beta_1, \beta_2\}, 
\{\gamma_1, \gamma_2\}, 
\{\sigma_1^2, \sigma_2^2\}
\}. 
$$
That is, we consider the three groups: the coefficients for sector-specific human capital,
the sector-specific benefits of staying in the same sector,
and 
sector-specific variances.

As an example, we set the benchmark vector of parameter values
at $\theta^b$ at the values used in \citet{honore2017poor}.
That is, 
\[
\beta_1^b = (1,1)^\top,
\beta_2^b = (0.5, 1)^\top,
\gamma_1^b = 0,
\gamma_2^b = 1,
(\sigma_1^2)^b = 2,
(\sigma_2^2)^b = 3.
\]
Suppose that we set  the counterfactual vector $\theta^c$ of parameter values at
\[
\beta_1^c = (1,2)^\top,
\beta_2^c = (0.5, 2)^\top,
\gamma_1^c = 0,
\gamma_2^c = 2,
(\sigma_1^2)^c = 2,
(\sigma_2^2)^c = 6.
\]
If we focus on income inequality, we may consider: 
\begin{itemize}
\item
between-sector inequality: $h_{\mathrm{bs-ineq}, t}(\theta) := \mathbb{E}[ w_{it} | d_{it} = 1] - \mathbb{E}[ w_{it} | d_{it} = 2]$;
\item
within-sector inequality:
$h_{\mathrm{ws-ineq}, t}(\theta; \tau_1, \tau_2) 
:= \mathbb{Q}_{w_{it}} (\tau_1| d_{it} = s ) - \mathbb{Q}_{w_{it}} (\tau_2| d_{it} = s )$,
where 
$\mathbb{Q}_{w_{it}} (\tau| d_{it} = s )$ is the $\tau$-quantile of $w_{it}$ conditional on $d_{it} = s$;
\item
overall inequality:
$h_{\mathrm{overall-ineq}, t}(\theta; \tau_1, \tau_2) 
:= \mathbb{Q}_{w_{it}} (\tau_1 ) - \mathbb{Q}_{w_{it}} (\tau_2 )$,
where 
$\mathbb{Q}_{w_{it}} (\tau )$ is the $\tau$-quantile of $w_{it}$.
\end{itemize}
As an example, we consider the change in overall inequality, measured by the $0.9-0.1$ quantile difference, from period 1 to period 2, that is:
$$
f(\theta) = 
h_{\mathrm{overall-ineq}, 2}(\theta; \tau_1 = 0.1, \tau_2 = 0.9) 
-
h_{\mathrm{overall-ineq}, 1}(\theta; \tau_1 = 0.1, \tau_2 = 0.9). 
$$
This example of $f(\theta)$ is not generally differentiable with respect to $\theta$. 
In consequence, the resulting $g(A)$ can not be viewed as the standard derivative; however, our approach provides how to quantify the importance of changes in the parameters.
Additionally, we offer a global sensitivity analysis, as opposed to a local sensitivity analysis.

As in Section~\ref{sec:toy}, we have that 
\begin{align*}
\bm{g} = 
\begin{pmatrix}
g (\{\beta_1, \beta_2\} ) \\
g ( \{\gamma_1, \gamma_2\} ) \\
g ( \{\sigma_1^2, \sigma_2^2\} ) \\
g (\{\beta_1, \beta_2\} \cup  \{\gamma_1, \gamma_2\} ) \\
g (\{\beta_1, \beta_2\} \cup  \{\sigma_1^2, \sigma_2^2\} ) \\
g ( \{\gamma_1, \gamma_2\} \cup  \{\sigma_1^2, \sigma_2^2\} ) \\
\end{pmatrix}, \;
\bm{D} =
\begin{pmatrix}
1 & 0 & 0 \\
0 & 1 & 0 \\
0 & 0 & 1 \\
1 & 1 & 0 \\
1 & 0 & 1 \\
0 & 1 & 1 \\
\end{pmatrix}, \;
\bm{\phi} = 
\begin{pmatrix}
\phi_{ \{\beta_1, \beta_2\} } \\
\phi_{ \{\gamma_1, \gamma_2\} } \\
\phi_{ \{\sigma_1^2, \sigma_2^2\} } 
\end{pmatrix},
\end{align*}
and 
$k ( \Pi, \Psi ) = 1$ for any $\Psi \subsetneqq \Pi$. 
Thus, $\bm{\phi}$ minimizes
\begin{align*}
( \bm{g} - \bm{D} \bm{\phi})^\top  ( \bm{g} - \bm{D} \bm{\phi})
\end{align*}
subject to 
$$
\phi_{\{\beta_1, \beta_2\} } + \phi_{ \{\gamma_1, \gamma_2\} } + \phi_{  \{\sigma_1^2, \sigma_2^2\} } = g( \{ \beta_1, \beta_2, \gamma_1, \gamma_2, \sigma_1^2, \sigma_2^2 \}).
$$
Again, this is a constrained least squares problem, which can be solved easily. 
The values of $\bm{g}$ are obtained by Monte Carlo simulation with the sample size of $10^7$.

Table~\ref{tab:shapley} gives the counterfactual simulation results. It can be seen that all three components contribute to the increase in inequality. However, 
$\phi_{ \{\gamma_1, \gamma_2\} }$ matters most in this example.
Recall that 
$(\gamma_1^b = 0, \gamma_2^b = 1)$
and
$(\gamma_1^c = 0, \gamma_2^c = 2)$.
In other words, sector 2 specific return to the log wages doubled in period 2, which explains the largest increase in the change in the overall inequality.

\begin{table}[htbp]
\caption{Shapley Value Decomposition}
\begin{center}
\begin{tabular}{ccc}
\hline\hline
Parameter & Value  &  Share \\
\hline
$\phi_{ \{\beta_1, \beta_2\} }$    &  182.5 &  0.30 \\
$\phi_{ \{\gamma_1, \gamma_2\} }$ & 264.4 & 0.43 \\
$\phi_{ \{\sigma_1^2, \sigma_2^2\} }$ & 164.2 & 0.27 \\
\hline
\end{tabular}
\end{center}
Notes. The share of the Shapley value is given by 
$\phi_{ \{ \cdot \} } / g( P )$.
\label{tab:shapley}
\end{table}



In short, we demonstrate that although a measure of income equality from a Roy model can be highly nonlinear in inputs, the group Shapley value generates unique additive contributions from $\Pi$. 
From a practical perspective, we find that Table~\ref{tab:shapley} is just as effortless to interpret as a regression table. Generally speaking, the group Shapley value provides a straightforward method to interpret counterfactual simulation results when
 a structural model simulates the output.  
 
\section{An Application to Capital Misallocation}\label{sec:example}

In this section, we revisit the work of \citet{DV2019}, who developed a method to disentangle sources of capital misallocation---specifically, the dispersion in average revenue products of capital (\emph{arpk}) using observable data on value-added and inputs from both China and United States.

\subsection{The Quantitative Framework of \citet{DV2019}}

For firm $i$ in period $t$, \citet{DV2019} define
the average revenue product of capital (\emph{arpk})  as 
$arpk_{it} = va_{it} - k_{it}$, where 
$va_{it}$ is the log value added and $k_{it}$ is the capital stock.
The parameter of interest is the variance $\sigma_{aprk}^2$ of $arpk_{it}$.
To simulate this quantity, we need to specify a variety of parameters. First of all, we need to choose the value of $\alpha$, which determines the curvature of operating profits.
There are two important dynamics in \citet{DV2019}. One is associated with the log productivity (denoted by $a_{it}$):
\begin{align}\label{a-process}
     a_{it} = \rho a_{it-1} + \mu_{it}, \; \mu_{it} \sim \mathcal{N}(0, \sigma_\mu^2),
\end{align}
where $\rho$ is the autoregressive parameter and $\sigma_\mu^2$ is the variance of the innovation term for productivity (see equation (5) in \citet{DV2019}). 
The other dynamics involves log distortion (denoted by $\tau_{it}$):
\begin{align}\label{tau-process}
    \tau_{it} = \gamma a_{it} + \varepsilon_{it} + \chi_i, \; 
    \varepsilon_{it} \sim \mathcal{N}(0, \sigma_\varepsilon^2),
    \chi_i \sim \mathcal{N}(0, \sigma_\chi^2),
\end{align}
where $\gamma$ determines how the distortion co-moves with productivity,
$\sigma_\varepsilon^2$ is the variance of i.i.d. shocks 
and 
$\sigma_\chi^2$ is the variance of time-invariant firm-specific shocks (see equation (6) in \citet{DV2019}).
The distribution of future productivity ($a_{it+1}$) conditional on the firm’s information set ($\mathcal{I}_{it}$) in period $t$
follows a normal distribution with the posterior mean $E_{it}[a_{it+1}]$ and the posterior variance $V$.
Finally, investment is subject to quadratic adjustment costs and the severity of the adjustment cost is parameterized via $\xi$.
In summary, there are 8 key parameters: $\alpha$, $\rho$, $\sigma_\mu^2$, $\xi$, $V$,
 $\gamma$, $\sigma_\varepsilon^2$, and $\sigma_\chi^2$ (see Table 1 in \citet{DV2019}). 

Recall that Table~\ref{tab:DV2019:intro} reproduces the parameter estimates from \citet{DV2019}.
The value of $\alpha$ is slightly larger for China.
The level of persistence for productivity is similar between the two countries, whereas the variance of productivity is larger in China than in U.S.
The adjustment costs are much higher in the U.S. but uncertainty is larger in China.
Note that $\gamma$ is more negative for China, indicating that the extent to which 
the distortion discourages investment by firms with higher productivity is larger in China than in U.S. Finally, in China, the variance of the transitory shocks is smaller but that of the permanent shocks is larger.

\subsection{\citet{DV2019}'s Decomposition}

\citet[Table 3]{DV2019} report the relative contributions of adjustment costs, uncertainty, and other factors to the \emph{arpk} dispersion under the assumption that only the factor of interest is operational. 
To be more specific, we let
$\theta = (\alpha, \rho, \sigma_\mu^2, \xi, V, \gamma, \sigma_\varepsilon^2,\sigma_\chi^2)$ and 
$f(\theta) = \sigma_{aprk}^2$. Assuming $f(\mathbf{0})=0$, \citet{DV2019} estimate the contribution of each factor by $\phi_{\mathrm{DV}} (A) := f(\theta_A, \mathbf{0}_{P \backslash A}) - f(\mathbf{0})$ where $\theta$ is the parameter estimates in Table~\ref{tab:DV2019:intro}.
The contribution estimation is based on the principle of \textit{ceteris paribus}, measuring the impact of each factor from a hypothetical void setting $f(\mathbf{0})$. However, while it intends to quantify the change when a factor of interest is added, $\phi_{\mathrm{DV}}$ can yield an unrealistic or imprecise assessment of individual impacts because setting some entries of $\theta$ at zeros results in degenerate dynamics. For instance, $V$ is the posterior variance, which needs to be strictly greater than zero.
Furthermore, $\phi_{\mathrm{DV}}$ does not satisfy the Efficiency axiom as it is a weighted sum of utility changes, which is referred to as a semivalue \citep{semivalue_without_efficiency}. This is why the relative contributions do not necessarily sum to $1$, as acknowledged by \citet{DV2019}. We address these two issues in the following subsection.




\subsection{Shapley Value Decomposition}

In this subsection, we carry out a new counterfactual exercise that is different from that of \citet[Table 3]{DV2019}. 
We set the benchmark vector ($\theta_b$) of parameter values using estimates from U.S. and the counterfactual vector ($\theta_c$) using estimates from China.
As each parameter represents a distinct component of the  model developed in \citet{DV2019}, we consider singleton groups:
\begin{align*}
\Pi = \{ 
\{\alpha\}, \{\rho\}, \{\sigma_\mu^2\}, 
\{\xi\}, \{V\}, \{\gamma\}, \{\sigma_\varepsilon^2\},
\{\sigma_\chi^2\}
\}.
\end{align*}
In the data, $\sigma_{aprk}^2$ was 0.92 for China and 0.45 for U.S (see Table 2 in \citet{DV2019}). Thus, the degrees of capital misallocation double by changing the U.S. parameter estimates to those based on Chinese data.
Our Shapley value decomposition will indicate the extent to which each parameter contributes to this increase.



\begin{table}[htbp]
\caption{Shapley Value Decomposition} 
\begin{center}
\begin{tabular}{ccc}
\hline\hline
Parameter & Value  &  Share \\
\hline
$\phi_{ \{\alpha\} }$  &   0.0073  & 0.016 \\
$\phi_{ \{\rho\} }$ & -0.0599 & -0.129 \\
$\phi_{ \{\sigma_\mu^2\} }$ & 0.1677 & 0.361 \\
$\phi_{ \{\xi\} }$ & -0.0588 & -0.126 \\
$\phi_{ \{V\} }$ & 0.0443 & 0.095 \\
$\phi_{ \{\gamma\} }$ & 0.2506 & 0.539 \\
$\phi_{ \{\sigma_\varepsilon^2\} }$ & -0.0061 & -0.013 \\
$\phi_{ \{\sigma_\chi^2\} }$ & 0.1200 & 0.258 \\
\hline
\end{tabular}
\end{center}
Notes. The share of the Shapley value is given by 
$\phi_{ \{ \cdot \} } / g( P )$.
\label{tab:shapley_DV}
\end{table}

Table~\ref{tab:shapley_DV} reports the resulting Shapley value decomposition.
First of all, note that there are positive and negative values for 
$\phi_{ \{ \cdot\} }$ because some changes in the parameter values imply increases in the \emph{arpk} dispersion but other changes indicate decreases. However, the sum of the Shapley shares is 1 by design. 

We now examine each row. The increase of $\alpha$ from $0.62$ to $0.71$ is associated with
$\phi_{ \{\alpha\} } = 0.0073$, implying the share of $0.016$, which is quite small.
The persistence parameter $\rho$ gets slightly smaller moving from the U.S. to China (that is, from $0.93$ to $0.91$). However, the resulting Shapley value is $-0.0599$ with a relatively large share of $-0.129$. We can interpret that a small decrease in the autoregressive parameter in the log productivity is associated with a relatively large reduction in the capital misallocation. 
The variance of the innovation term for productivity increases more substantially from 0.08 to 0.15, resulting in the second largest share in the Shapley value decomposition. It is interesting to note that the change in the adjustment costs looks more visible in the sense that $\xi$ decreases by a factor of 10, while the implied share is only $-0.126$. This shows that the Shapley value decomposition is useful to compare changes in parameters on the same scale. The increase in uncertainty is associated with the share of $0.095$. The largest Shapley share  of $0.542$ comes from $\gamma$, which is negative for both countries. A more negative $\gamma$ means that the distortion plays a larger role in disincentivizing investment by more productive firms. The reduction in the variance of the transitory shocks is small and as a result,  the Shapley share is small as well. The increase in the variance of the permanent shocks is quite large, ranking as the third most important factor in terms of the Shapley share. In short, the two components ($\gamma$ and $\sigma_\chi^2$) of the distortion as well as the variance of the productivity shocks are the three most important factors affecting the increase in the \emph{arpk} dispersion. Overall, our Shapley value analysis reconfirms the quantitative findings in \citet{DV2019}, while providing novel perspectives.

\section{An Application to Globalization}\label{sec:example:more}

\citet{CGT2016} examine decreases in trade barriers, tariffs, and firing costs in an open economy, using establishment-level data from Colombia. Their counterfactual experiments suggest that Colombia's integration into global product markets raised its national income but also increased unemployment, wage inequality, and firm-level volatility.
Table 4 in \citet{CGT2016} provides the results of their counterfactual experiments.
Suppose that we are interested in comparing the last column ``Reforms and globalization'' with the first column ``Baseline'' and construct Shapley value decomposition among three elements: firing cost ($c_f$), tariff rate ($\tau_a$), and iceberg trade cost ($\tau_c$). 
Then, we have the same structure as the toy example in Section~\ref{sec:toy}. That is, we have that 
\begin{align*}
\bm{g} = 
\begin{pmatrix}
g ( \{c_f\} ) \\
g ( \{\tau_a\} ) \\
g ( \{\tau_c\} ) \\
g (\{c_f\} \cup  \{\tau_a\} ) \\
g (\{c_f\} \cup  \{\tau_c\} ) \\
g ( \{\tau_a\} \cup  \{\tau_c\} ) \\
\end{pmatrix}, \;
\bm{D} =
\begin{pmatrix}
1 & 0 & 0 \\
0 & 1 & 0 \\
0 & 0 & 1 \\
1 & 1 & 0 \\
1 & 0 & 1 \\
0 & 1 & 1 \\
\end{pmatrix}, \;
\bm{\phi} = 
\begin{pmatrix}
\phi_{ \{c_f\} } \\
\phi_{ \{\tau_a\} } \\
\phi_{ \{\tau_c\} } 
\end{pmatrix}.
\end{align*}
For self-containment, we re-produce the parameter values in Panel A of Table~\ref{CGT-tab4}.
In Panel A, the benchmark (``Baseline'') parameters are $(c_f^b, \tau_a^b, \tau_c^b) = (0.60, 1.21, 2.50)$, whereas the counterfactual (``Reforms and globalization'') parameters are $(c_f^c, \tau_a^c, \tau_c^c) = (0.30, 1.11, 2.19)$. 
Regarding $g(\cdot)$, we consider the aggregates reported in Table 4 in \citet{CGT2016}, which are reproduced in Panel B of Table~\ref{CGT-tab4}.
As before, $g(\cdot)$ is obtained by taking the difference between the counterfactual (``Reforms and globalization'') quantities and the benchmark (``Baseline'') quantities, where the latter quantities are normalized to be one.

It can be seen from Table~\ref{CGT-tab4} that the four middle columns are the input for optimization (i.e. elements of $\bm{g}$): using the notation above, we have values for 
$g ( \{c_f\} )$,
$g ( \{\tau_a\} )$,
$g ( \{\tau_c\} )$, 
and
$g (\{c_f\} \cup  \{\tau_a\} )$. They are called ``Labor'', ``Tariff'', ``Iceberg'', and ``Reforms'' in the table.
However, the last two elements of $\bm{g}$ are missing: namely,
$g (\{c_f\} \cup  \{\tau_c\} )$
and
$g ( \{\tau_a\} \cup  \{\tau_c\} )$.
It would be ideal to complete the missing elements in order to construct Shapley value decomposition. However, one interesting research question is what one can do if it is expensive or impossible to carry out further experiments to obtain the missing elements for $\bm{g}$.

In view of the fact that each of the parameters is reduced to favor globalization, we consider the following set of restrictions on the missing elements:
for $g (\{c_f\} \cup  \{\tau_c\} )$, 
\begin{align}\label{ls-1}
\begin{split}
 g ( \{c_f\} ) \cdot \mathrm{sgn} [g ( \{\tau_c\} )] &\leq g (\{c_f\} \cup  \{\tau_c\} ), \\
g ( \{\tau_c\} ) \cdot \mathrm{sgn} [g ( \{c_f\} )] &\leq g (\{c_f\} \cup  \{\tau_c\} ), \\
g (\{c_f\} \cup  \{\tau_c\} ) \cdot \mathrm{sgn} [g ( \{\tau_a\} )]  &\leq g (\{c_f\} \cup  \{\tau_a\} \cup \{\tau_c\}), 
\end{split}
\end{align}
and
for $g (\{d_a\} \cup  \{\tau_c\} )$, 
\begin{align}\label{ls-2}
\begin{split}
g ( \{d_a\} ) \cdot \mathrm{sgn} [g ( \{\tau_c\} )] &\leq g (\{d_a\} \cup  \{\tau_c\} ), \\
g ( \{\tau_c\} ) \cdot \mathrm{sgn} [g ( \{d_a\} )] &\leq g (\{d_a\} \cup  \{\tau_c\} ), \\
g (\{d_a\} \cup  \{\tau_c\} ) \cdot \mathrm{sgn} [g ( \{\tau_a\} )]  &\leq g (\{c_f\} \cup  \{\tau_a\} \cup \{\tau_c\}).
\end{split}
\end{align}
In addition, suppose that we have known upper and lower bounds on 
$g (\{c_f\} \cup  \{\tau_c\} )$ 
and 
$g (\{d_a\} \cup  \{\tau_c\} )$:
\begin{align}\label{ls-3}
    g_{\min} \leq g (\{c_f\} \cup  \{\tau_c\} ) \leq g_{\max} \; \text{ and } \; g_{\min} \leq g (\{d_a\} \cup  \{\tau_c\} ) \leq g_{\max},
\end{align}
where $g_{\min}$ and $g_{\max}$ are predetermined constants such that $g_{\min} = 0.5$ and $g_{\max} = 1.5$.
Combining all the constraints above in \eqref{ls-1}-\eqref{ls-3} forms the linear constraints $A_{\mathrm{const}} \bm{g} \leq b_{\mathrm{const}}$ in Section~\ref{sec:SB-SMNS}.

\begin{table}[htb]
\caption{Counterfactual experiments from \citet{CGT2016}}\label{CGT-tab4}
\begin{tabular}{lcccccc}
\hline\hline 
& Baseline & Labor & Tariff & Iceberg & Reforms & \begin{tabular}{l} 
Reforms and \\
globalization
\end{tabular} \\
\hline 
\multicolumn{7}{l}{Panel A: Parameters} \\
$c_f$ (firing cost) & 0.60 & 0.30 & 0.60 & 0.60 & 0.30 & 0.30 \\
$\tau_a$ (ad valorem tariff rate) & 1.21 & 1.21 & 1.11 & 1.21 & 1.11 & 1.11 \\
$\tau_c$ (iceberg trade cost) & 2.50 & 2.50 & 2.50 & 2.19 & 2.50 & 2.19 \\
& & & & & & \\
\hline \multicolumn{7}{l}{Panel B: Aggregates} \\
Revenue share of exports & 1.00 & 1.02 & 1.36 & 2.01 & 1.39 & 2.50 \\
Exit rate & 1.00 & 0.96 & 1.02 & 1.13 & 0.96 & 1.03 \\
 Job turnover & 1.00 & 0.90 & 1.01 & 1.03 & 0.92 & 0.94 \\
 Mass of firms & 1.00 & 0.96 & 0.95 & 0.74 & 0.88 & 0.66 \\
Share of labor, $Q$ sector & 1.00 & 1.06 & 1.01 & 0.88 & 1.07 & 0.98 \\
Vacancy filling rate $(\phi)$ & 1.00 & 0.93 & 1.04 & 1.11 & 0.99 & 1.09 \\
Unemp. rate, $Q$ sector & 1.00 & 0.73 & 1.11 & 1.38 & 0.88 & 1.19 \\
Std. wages (firms) & 1.00 & 1.09 & 1.01 & 1.03 & 1.12 & 1.18 \\
 Std. wages (workers) & 1.00 & 1.10 & 1.02 & 1.04 & 1.11 & 1.14 \\
Std. $J$ (firms) & 1.00 & 1.05 & 1.03 & 1.06 & 1.09 & 1.18 \\
 Std. $J$ (workers) & 1.00 & 1.04 & 1.04 & 1.06 & 1.07 & 1.21 \\
 Exchange rate & 1.00 & 1.04 & 0.99 & 0.89 & 1.02 & 0.84 \\
 Real income & 1.00 & 0.95 & 1.00 & 1.14 & 0.96 & 1.12 \\
\hline
\end{tabular}
\parbox{6in}{
Notes. This table is reproduced from Table 4 in \citet{CGT2016}. Aggregate statistics in Panel B are normalized by their baseline levels.}
\end{table}

\begin{table}[htbp]
\caption{Shapley Value Decomposition}\label{tab:shapley_Cosar}
\begin{center}
\begin{tabular}{lcccccc}
\hline\hline
Aggregates & & $\phi_{ \{c_f\} }$  &  $\phi_{ \{\tau_a\} }$ & $\phi_{ \{\tau_c\} }$\\
\hline
\multirow{3}{*}{Revenue share of exports} &SLB &   0.012 & 0.182 & 0.98 \\
& SUB & 0.257 & 0.427 & 1.143 \\
& SMNS & 0.175 & 0.345 & 0.98 \\
\hline
\multirow{3}{*}{Exit rate} &SLB & -0.542 & 0.007  & 0.085\\
& SUB & -0.073 & 0.258 & 0.325 \\
& SMNS & -0.075 &  0.01 &   0.095 \\
\hline
\multirow{3}{*}{Mass of firms} &SLB & -0.538 & -0.543 & -0.032\\
& SUB & 0.042 & 0.037 & 0.355 \\
& SMNS & -0.152 & -0.157 & -0.032 \\
\hline
\multirow{3}{*}{Vacancy filling rate ($\phi$)} &SLB & -0.532 & 0.023 & 0.082\\
& SUB & -0.042 & 0.308 & 0.34 \\
& SMNS & -0.045 & 0.03 &  0.105 \\
\hline
\multirow{3}{*}{Unemp. rate, Q sector} &SLB & -0.673 & 0.062  & 0.275\\
& SUB & -0.223 & 0.402 & 0.538 \\
& SMNS & -0.224 & 0.063 & 0.351 \\
\hline
\multirow{3}{*}{Std. wage (firms)} &SLB & 0.058 & 0.008 & 0.033\\
& SUB & 0.123 & 0.063 & 0.073 \\
& SMNS & 0.06 & 0.06 & 0.06 \\
\hline
\multirow{3}{*}{Std. wage (workers)} &SLB & 0.058 & 0.008 & 0.027\\
& SUB & 0.098 & 0.038 & 0.05 \\
& SMNS & 0.058 & 0.038 & 0.043 \\
\hline
\multirow{3}{*}{Std. $J$ (firms)} &SLB & 0.027 & 0.017 & 0.057\\
& SUB & 0.087 & 0.077 & 0.097 \\
& SMNS & 0.06 & 0.06 & 0.06 \\
\hline
\multirow{3}{*}{Std. $J$ (workers)} &SLB & 0.018 & 0.018 & 0.073\\
& SUB & 0.093 & 0.093 & 0.123 \\
& SMNS & 0.068 & 0.068 & 0.073 \\
\hline
\multirow{3}{*}{Real income} &SLB & -0.515 & -0.482 & 0.1 \\
& SUB & 0.243 & 0.285 & 0.608 \\
& SMNS & -0.015 & 0.035 & 0.1 \\
\hline
\end{tabular}
\end{center}
\parbox{6in}{
Notes. SLB, SUB, and SMNS denotes Shapley Lower Bounds, Shapley Upper Bounds, and Shapley Minimum Norm Solution, respectively. Linear constraints of the three aggregate statistics: job turnover, share of labor, and exchange rate are incompatible to each other, and thus they are excluded.}
\end{table}

\begin{figure}[htb]
\caption{Shapley Minimum Norm Solution with Shapley Bounds}\label{fig:SMNS}
\begin{center}
\includegraphics[scale=0.7]{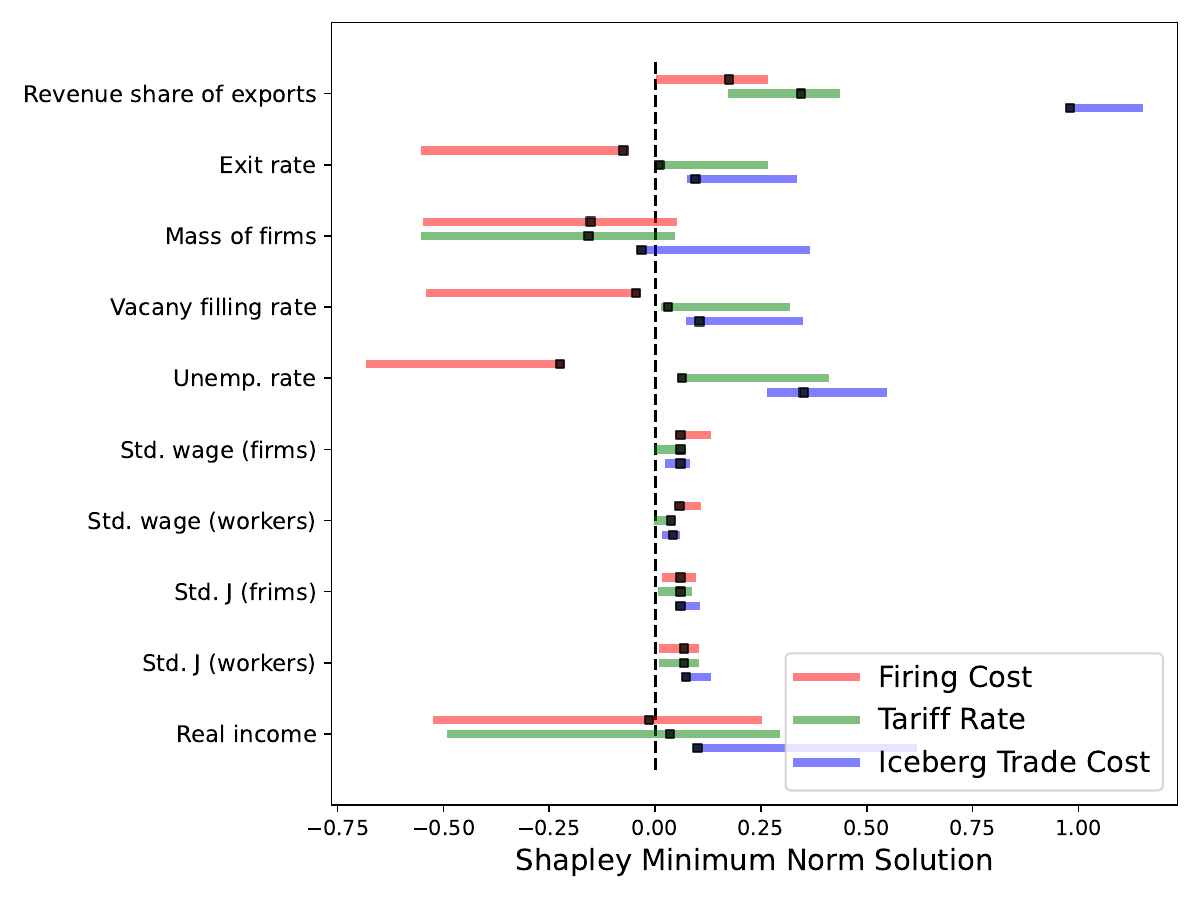}    
\end{center}
\parbox{6in}{
Notes. Each interval corresponds to Shapley bounds, while the black square represents the Shapley minimum norm solution.}
\end{figure}

Table~\ref{tab:shapley_Cosar} and Figure~\ref{fig:SMNS} present the empirical results for the Shapley bounds and the Shapley minimum norm solutions using the methodology developed in Section \ref{sec:SB-SMNS}. As can be seen from Table~\ref{CGT-tab4}, the revenue share of exports increases by 1.5 from the benchmark quantity to the counterfactual quantity. 
The largest portion of this increase is explained by the decrease in the iceberg trade cost (with a lower bound of 0.98).
The reduction in the tariff rate appears to explain slightly more than the reduction in firing costs, although their bounds overlap, making it difficult to definitively rank these two factors. 
Whether considering the Shapley bounds or the Shapley minimum norm solutions, firing cost reductions are linked to decreases in the exit rate, vacancy filling rate, and unemployment rate. Conversely, reductions in the tariff rate and iceberg trade cost are associated with increases in these rates. The decomposition result for the mass of firms is ambiguous: for all three factors, the minimum norm solutions are negative, but the bounds include zero. All the inequality measures---Std. wages (firms), Std. wages (workers),  Std. $J$ (firms), and Std. $J$ (workers)---increase with the reduction in each factor. However, there is no definite ranking among the three factors in terms of their contribution to increases in inequality. 
Finally, real income increases by 0.12 from the benchmark quantity to the counterfactual quantity, primarily explained by the reduction in the iceberg trade cost (with a lower bound of 0.1). The Shapley bounds are rather wide for both the firing cost and the tariff rate.  

Overall, the empirical results suggest that the reduction in iceberg trade costs is the most significant factor for real income growth. However, it is less clear which factor is the primary driver of increases in inequality. Returning to our research question, this example demonstrates that it is feasible to perform a coherent Shapley value decomposition even when it is impossible to conduct further experiments to obtain missing elements for $\bm{g}$. Thus, our methodology holds potential for real-world applications that involve costly counterfactual simulations.


\section{Conclusions}\label{sec:conclusions}


Although Shapley values are widely used across various disciplines, their application to interpreting counterfactual simulations in structural economic models has not yet been explored. We have demonstrated that the group Shapley value provides a natural framework for quantifying the importance of parameter changes in complex models. By preserving the axiomatic properties of the Shapley value, our method offers a well-grounded decomposition approach for explaining outcomes in counterfactual simulations. Practically, it enables researchers to generate importance tables that are as intuitive and accessible as regression tables. In this sense, our use of the Shapley value could pave the way for ``interpretable structural economics,'' much like \citet{LundbergLee}'s introduction of Shapley values to the machine learning community has advanced explainable artificial intelligence.

However, our method has some limitations. First, we focus on comparing two sets of parameters. Extending our approach to cases involving more than two sets of parameter values (e.g., the United States, China, and South Korea, using the example in the introduction) would be valuable. Second, our numerical examples are limited to small-scale optimization problems, where computational complexity is less of a concern. For larger-scale problems, computational challenges become more significant, and developing strategies to address them—such as devising effective sampling techniques and accounting for sampling errors—will be crucial. These are exciting directions for future research.


\appendix

\section{Appendix}\label{sec:appendix}

\begin{proof}[Proof of Proposition~\ref{opt-result}]
We start with the Lagrangian problem:
\begin{align*}
\sum_{\Psi \subsetneqq \Pi}  
\left\{ 
g\left( \bigcup_{A \in \Psi}  A \right)  -  \sum_{A \in \Psi} \phi_A 
\right\}^2 k ( \Pi, \Psi ) + 2 \lambda \left[ \sum_{A \in \Pi} \phi_A - g \left(  P \right) \right].
\end{align*}
The resulting optimal solutions $\{ \phi_M^* : M \in \Pi  \}$ and $\lambda^*$ satisfy:
\begin{align*}
&\lambda^* = \sum_{\Psi \subsetneqq \Pi, M \in \Psi}  
\left\{ 
g\left( \bigcup_{A \in \Psi}  A \right)  -  \sum_{A \in \Psi} \phi_A^* \right\} k ( \Pi, \Psi ) \ \ \text{for each $ \phi_M^*$}, \\
&\sum_{A \in \Pi} \phi_A^*  =  g \left(  P \right). 
\end{align*}
For each integer $1 \leq j \leq |\Pi|-1$, let 
\begin{align*}
k ( |\Pi|, j ) := 
\begin{pmatrix}
|\Pi| - 2 \\
j - 1
\end{pmatrix}^{-1}.
\end{align*}
Define
\begin{align}
\mu_j (M) &:= k ( |\Pi|, j ) \sum_{| \Psi | = j, M \in \Psi}   g\left( \bigcup_{A \in \Psi}  A \right), \label{mu_j_def} \\
x^*_j (M) &:= k ( |\Pi|, j ) \sum_{| \Psi | = j, M \in \Psi}   \sum_{A \in \Psi} \phi_A^* .   \label{x_j_def}
\end{align}
For each $M \in \Pi$, write
\begin{align}
\lambda^* 
&= \sum_{j=1}^{|\Pi|-1} \sum_{| \Psi | = j, M \in \Psi} \left\{  g\left( \bigcup_{A \in \Psi}  A \right)  -  \sum_{A \in \Psi} \phi_A^* \right\} k ( |\Pi|, j ) \nonumber \\
&= \sum_{j=1}^{|\Pi|-1} \{ \mu_j (M) - x^*_j (M) \}. \label{lambda-eq}
\end{align}
Note that for each $M \in \Pi$, 
\begin{align*}
\sum_{| \Psi | = j, M \in \Psi}   \sum_{A \in \Psi} \phi_A^*
&= 
\phi_M^*
\begin{pmatrix}
|\Pi| - 1 \\
j - 1
\end{pmatrix}
+
\sum_{A \in \Psi \setminus \{ M \}} \phi_A^*
\begin{pmatrix}
|\Pi| - 2 \\
j - 2
\end{pmatrix}.
\end{align*}
Using this, further write 
\begin{align*}
\sum_{j=1}^{|\Pi|-1} x^*_j (M)
&= k ( |\Pi|, 1 ) \phi_M^*
+ \sum_{j=2}^{|\Pi|-1} k ( |\Pi|, j ) \sum_{| \Psi | = j, M \in \Psi}   \sum_{M \in \Psi} \phi_M^* \\
&= \sum_{j=1}^{|\Pi|-1} k ( |\Pi|, j ) 
\begin{pmatrix}
|\Pi| - 1 \\
j - 1
\end{pmatrix}
\phi_M^* 
+
\sum_{j=2}^{|\Pi|-1} k ( |\Pi|, j )  \begin{pmatrix}
|\Pi| - 2 \\
j - 2
\end{pmatrix}
\sum_{A \in \Psi \setminus \{ M \}} \phi_A^* \\
&= \sum_{j=1}^{|\Pi|-1} k ( |\Pi|, j ) 
\begin{pmatrix}
|\Pi| - 1 \\
j - 1
\end{pmatrix}
\phi_M^* 
- 
\sum_{j=2}^{|\Pi|-1} k ( |\Pi|, j )  \begin{pmatrix}
|\Pi| - 2 \\
j - 2
\end{pmatrix}
\phi_M^* \\
&\;\;\; +
\sum_{j=2}^{|\Pi|-1} k ( |\Pi|, j )  \begin{pmatrix}
|\Pi| - 2 \\
j - 2
\end{pmatrix}
\sum_{A \in \Psi} \phi_A^* \\
&= \sum_{j=1}^{|\Pi|-1} k ( |\Pi|, j ) 
\begin{pmatrix}
|\Pi| - 2 \\
j - 1
\end{pmatrix}
\phi_M^* 
+ 
\sum_{j=2}^{|\Pi|-1} k ( |\Pi|, j )  \begin{pmatrix}
|\Pi| - 2 \\
j - 2
\end{pmatrix}
g \left(  P \right) \\
&= 
( |\Pi| - 1 )
\phi_M^* 
+ 
\sum_{j=2}^{|\Pi|-1} k ( |\Pi|, j )  \begin{pmatrix}
|\Pi| - 2 \\
j - 2
\end{pmatrix}
g \left(  P \right),
\end{align*}
where the last equality follows from the fact that 
\begin{align*}
k ( |\Pi|, j ) 
\begin{pmatrix}
|\Pi| - 2 \\
j - 1
\end{pmatrix}
&= 1.
\end{align*}
Therefore, it follows from \eqref{lambda-eq} that 
\begin{align}\label{lambda-eq-new}
( |\Pi| - 1 )
\phi_M^* 
+ 
\sum_{j=2}^{|\Pi|-1} k ( |\Pi|, j )  \begin{pmatrix}
|\Pi| - 2 \\
j - 2
\end{pmatrix}
g \left(  P \right)
&=
\sum_{j=1}^{|\Pi|-1} \mu_j (M)
- \lambda^*.
\end{align}
Now summing  \eqref{lambda-eq-new} over $M \in \Pi$ yields
\begin{align*}
 \lambda^* 
&= \frac{1}{|\Pi|} 
\left[ 
\sum_{M \in \Pi} \sum_{j=1}^{|\Pi|-1} \mu_j (M) 
- ( |\Pi| - 1 ) g(P) 
\right]
- \sum_{j=2}^{|\Pi|-1} k ( |\Pi|, j )  \begin{pmatrix}
|\Pi| - 2 \\
j - 2
\end{pmatrix}
g \left(  P \right).
\end{align*}
Thus, it follows from \eqref{lambda-eq-new} that 
\begin{align}\label{gsv-der}
\phi_M^* 
&=
\frac{1}{|\Pi| - 1}
\left[
\sum_{j=1}^{|\Pi|-1} \mu_j (M)
- 
\frac{1}{|\Pi|}
\left\{
\sum_{M \in \Pi} \sum_{j=1}^{|\Pi|-1} \mu_j (M) 
- ( |\Pi| - 1 ) g(P) 
\right\}
\right].
\end{align}
To prove the result, it remains to show that for any set function $g \in \mathbb{R}^{\Pi}$, \eqref{gsv-der} is identical to the following expression:
%
%
\begin{align}\label{gsv-alt-simple}
\begin{split}
\frac{1}{|\Pi|}
\sum_{\Psi \subset \Pi} 
\begin{pmatrix}
|\Pi| - 1 \\
|\Psi| - 1
\end{pmatrix}^{-1}
\left[ 
g\left( \bigcup_{A \in \Psi} A \right)
-
g\left( \bigcup_{A \in \Psi \setminus \{ M \}} A \right)
\right].
\end{split}
\end{align} 
To this end, we now consider a function $g_{\Upsilon}(\cdot)$ such that
\begin{align*}
g_{\Upsilon}\left( \bigcup_{A \in \Upsilon} A \right) = 1
\ \ \text{ and } \ \
g_{\Upsilon}\left( \bigcup_{A \in \Psi} A \right) = 0 \text{ for any $\Psi \neq \Upsilon$}.
\end{align*}
Since $\{g_{\Upsilon}: \Upsilon \subset \Pi \}$ is an orthonormal basis of $\mathbb{R}^{\Pi}$, it is enough to show the identity for every $g_{\Upsilon}$.
It follows from \eqref{gsv-alt-simple} that 
\begin{align}\label{phi-optimal-1}
\begin{split}
\phi_M^* =
\left\{ 
\begin{array}{cc}
\frac{1}{|\Pi|}
\begin{pmatrix}
|\Pi| - 1 \\
|\Upsilon| - 1
\end{pmatrix}^{-1}
& \text{ if $M \in \Upsilon$}, \\
- \frac{1}{|\Pi|}
\begin{pmatrix}
|\Pi| - 1 \\
|\Upsilon|
\end{pmatrix}^{-1}
& \text{ if $M \notin \Upsilon$}.   
\end{array}
\right.
\end{split}
\end{align} 
By \eqref{mu_j_def}, we have that
\begin{align*}
\sum_{j=1}^{|\Pi|-1} \mu_j (M) 
&= \sum_{j=1}^{|\Pi|-1} k ( |\Pi|, j ) \sum_{| \Psi | = j, M \in \Psi}   g\left( \bigcup_{A \in \Psi}  A \right) \\
&=
\left\{ 
\begin{array}{cc}
k ( |\Pi|, |\Upsilon| )   & \text{ if $M \in \Upsilon$}, \\
0 & \text{ if $M \notin \Upsilon$}.  
 \end{array}
\right.
\end{align*}
Furthermore, note that $g_{\Upsilon}(P) = 0$ and 
\begin{align*}
\sum_{M \in \Pi} \sum_{j=1}^{|\Pi|-1} \mu_j (M) 
&= |\Upsilon| k ( |\Pi|, |\Upsilon| ). 
\end{align*}
Now use \eqref{gsv-der} to write 
\begin{align*}
\phi_M^* 
&=
\frac{1}{|\Pi| - 1}
\left[
\sum_{j=1}^{|\Pi|-1} \mu_j (M)
- 
\frac{1}{|\Pi|}
\left\{
\sum_{M \in \Pi} \sum_{j=1}^{|\Pi|-1} \mu_j (M) 
- ( |\Pi| - 1 ) g(P) 
\right\}
\right] \\
&=
k ( |\Pi|, |\Upsilon| )
\left\{ 
\begin{array}{cc}
\frac{1}{|\Pi| - 1} 
\left( 
1 - \frac{1}{|\Pi|} |\Upsilon| \right)
  & \text{ if $M \in \Upsilon$}, \\
- \frac{1}{|\Pi|(|\Pi| - 1)} 
 |\Upsilon| 
& \text{ if $M \notin \Upsilon$},
\end{array}
\right. 
\end{align*}
which is identical to those in \eqref{phi-optimal-1}. Therefore, we have proved the proposition.
\end{proof}

\bibliographystyle{chicago}
\bibliography{Shapley}

\end{document}